\documentclass[a4paper,11pt]{article}
\usepackage{cite}
\usepackage[all]{xy}
\usepackage{amsmath,amsfonts,amssymb,amsthm,epsfig,amscd,comment,latexsym,psfrag}
\usepackage{CJK}
\usepackage{mathrsfs}
\usepackage{nicefrac,xspace,tikz}
\usepackage{arydshln}
\usepackage{stmaryrd}
\usepackage{graphicx,bm}
\usepackage{amscd}
\usetikzlibrary{arrows}
\newtheorem{theorem}{Theorem}

\allowdisplaybreaks

\begin{document}
\begin{center}
{\Large\bf Higher order constraints for the ($\beta$-deformed) Hermitian matrix models}\vskip .2in
{\large Rui Wang$^{a,}$\footnote{wangrui@cumtb.edu.cn}} \vskip .2in
$^a${\em Department of Mathematics, China University of Mining and Technology,
Beijing 100083, China}\\

\begin{abstract}
We construct the ($\beta$-deformed) higher order total derivative operators and analyze their remarkable properties.
In terms of these operators, we derive the higher order constraints for the ($\beta$-deformed) Hermitian matrix models.
We prove that these ($\beta$-deformed) higher order constraints are reducible to the Virasoro constraints.
Meanwhile, the Itoyama-Matsuo conjecture for the constraints of the Hermitian matrix model is proved.
We also find that through rescaling variable transformations, two sets of the constraint operators become the $W$-operators of $W$-representations
for the ($\beta$-deformed) partition function hierarchies in the literature.
\end{abstract}

\end{center}

{\small Keywords: Matrix Models, Conformal and $W$ Symmetry}


\section{Introduction}

Matrix models usually satisfy infinite sets of constraints, which can be formulated in a form of differential equations with respect to the time variables.
These constraints can be considered as equations of motion in the general string theory associated with the model.
For the Hermitian matrix model
\begin{equation}\label{PF}
Z=\int d^Nz\ \ \Delta(z)^2 e^{\sum_{i=1}^{N}V(z_i)},
\end{equation}
where $\Delta(z)=\prod\limits_{1\leq i<j\leq N}(z_i-z_j)$ is the Vandermonde determinant,
$V(z)=\sum_{m=0}^{\infty}t_mz^m$,
there are the well-known Virasoro constraints \cite{Mironov,David,Ambjorn,Itoyama255,Dijkgraaf,Marshakov1991}
\begin{equation}
\hat L_nZ=0,\ \ n\geq -1,
\end{equation}
where
\begin{equation}\label{virconop}
\hat L_n=\sum_{k=1}^{\infty}kt_k\frac{\partial}{\partial t_{k+n}}+
\sum_{k=0}^{n}\frac{\partial}{\partial t_{k}}
\frac{\partial}{\partial t_{n-k}}.
\end{equation}
They may be derived from the equality
$\int d^Nz\ \ L_n\Delta(z)^2 e^{\sum_{i=1}^{N}V(z_i)}=0$,
where $L_n=\sum_{i=1}^{N}\frac{\partial}{\partial z_i}z_i^{n+1}$.

Itoyama and Matsuo \cite{Itoyama} proposed an approach to derive a large class of constraints, namely $W_{1+\infty}$ constraints,
which are associated with the higher order differential operators of the $W_{1+\infty}$ algebra.
More precisely, by inserting the $W_{1+\infty}$ operators $\mathcal{D}_{n+r,r}=\sum_{i=1}^{N}z_i^{n+r}\frac{\partial^r}{\partial z_i^r}$ and
$\mathcal{D}^{\dag}_{n+r,r}=(-1)^r\sum_{i=1}^{N}\frac{\partial^r}{\partial z_i^r}z_i^{n+r}$ ($r, n+r\in \mathbb{Z}_{+}$), the derived $W_{1+\infty}$ constraints are
\begin{equation}\label{dnm}
\tilde W_n^rZ=\int d^Nz\ \ \Delta \cdot\mathcal{D}_{n+r,r}(e^V\Delta)
-(\mathcal{D}^{\dag}_{n+r,r}\Delta)\cdot e^V\Delta =0.
\end{equation}
However, it is rather nontrivial to write down the constraints explicitly.
Itoyama and Matsuo conjectured that the constraint operators $\tilde W_n^r$ with $r\geq 2$
are reducible to the Virasoro constraint operators \cite{Itoyama}.
The modified techniques were applied to the Kontsevich model partition functions \cite{Khviengia}.
It was shown that the super eigenvalue model satisfies an infinite set of constraints with even spins which are associated with the bosonic generators of
the super ($W_{\frac{\infty}{2}}\oplus W_{\frac{1+\infty}{2}}$) algebra, and the simplest constraints ($s=4$) are reducible to the super Virasoro constraints \cite{Buffon}.

$W$-representations of matrix models realize the partition
functions by acting on elementary functions with exponents of the given $W$-operators
\cite{Shakirov2009,Kon-Wit,BGW,Cassia2020,MirGKM}.
Recently, the partition function hierarchies were presented by the expansions with respect to the symmetric functions via the $W$-representations \cite{wangliu,qtwlzz,swlzz,gwlzz}. The partition function hierarchies expanded by the Schur
functions can be described by the interpolating two-matrix model \cite{Alexandrov23,2301.04107,2301.11877}.
For the $\beta$-deformed partition function hierarchies in \cite{wangliu}, their integral realizations and Ward identities were presented
by means of  $\beta$-deformed Harish-Chandra-Itzykson-Zuber integral \cite{Oreshina1,Oreshina2,Oreshina3}.
For the generalized $\beta$ and $(q,t)$-deformed partition function hierarchies, it was found that there are the profound interrelations between them and the $4d$ and $5d$ Nekrasov partition functions \cite{gwlzz}.

The Hermitian matrix models can be represented as the integrated conformal field theory expectation values.
Recently, by constructing the operators in terms of the generators of the Heisenberg algebra and inserting them into the integrated expectation values,
the new constraints were presented \cite{Wangr2022}, where the $A$ and $B$-type Lassalle constraints are contained in them. It was shown that
these new constraints can be derived by inserting the second order total derivative operators
\begin{equation}\label{Mnz}
\bar W_{n}=\sum_{i=1}^{N}\frac{\partial^2}{\partial z^2_i}z^{n+2}_i
-2\sum_{i=1}^{N}\sum_{j\neq i}\frac{\partial}{\partial z_i}\frac{z^{n+2}_i}{z_i-z_j}
-(n+2)\sum_{i=1}^{N}\frac{\partial}{\partial z_i}z^{n+1}_i, \ \ n\geq -2,
\end{equation}
into the integrand of (\ref{PF}).
The interesting property of these constraint operators is that via rescaling variable transformations,
they give the $W$-operators of $W$-representations of some well-known matrix models, such as
the Gaussian Hermitian matrix model (in the external field) \cite{Shakirov2009,AMironov1705}, $N \times N$ complex
\cite{AlexandrovJHEP2009,AMironov1705} and Hurwitz-Kontsevich \cite{Goulden97,KHurwitz,Shakirov2009} matrix models.

In this paper, we will make a further step to investigate the constraints for the ($\beta$-deformed) Hermitian matrix models.
The goal of this paper is to construct the higher order constraints for these matrix models and
prove Itoyama-Matsuo conjecture.

This paper is organized as follows. In section 2, we construct the higher order total derivatives. By means of these operators, we derive the higher order constraints for the Hermitian matrix model. In terms of these constraints, we prove the Itoyama-Matsuo conjecture. We also point out the intrinsic connection between one set of the constraint operators and the $W$-operators of $W$-representations
for the partition function hierarchies in the literature.
In section 3, the higher order constraints for the $\beta$-deformed Hermitian matrix model
which are reducible to the Virasoro constraints are constructed. We find that one set of the derived
constraint operators are associated with the $W$-operators of $W$-representations
for the $\beta$-deformed partition function hierarchies in the literature.
We end this paper with the conclusion in section 4.

\section{Higher order constraints for the Hermitian matrix model}

Let us start with the second order total derivative operator $\bar W_0$ in (\ref{Mnz})
\begin{equation}\label{Mnz0}
\bar W_{0}=\sum_{i=1}^{N}\frac{\partial^2}{\partial z^2_i}z^{2}_i
-2\sum_{i=1}^{N}\sum_{j\neq i}\frac{\partial}{\partial z_i}\frac{z^{2}_i}{z_i-z_j}
-2\sum_{i=1}^{N}\frac{\partial}{\partial z_i}z_i.
\end{equation}

We construct a set of commutative operators in terms of (\ref{Mnz0})
\begin{equation}\label{Hnm}
H_n^{(m)}=\frac{(-1)^{n-1}}{(n-1)!}{\rm ad}^{n-1}_{F_{m+1}}F_{m},\ \ \ n,m \in \mathbb{Z}_{+},
\end{equation}
where $F_{m}={\rm ad}^{m-1}_{-\frac{1}{2}\bar {W}_0}{F}_{1}$ and
$F_1=L_{-1}=\sum_{i=1}^{N}\frac{\partial}{\partial z_i}$.

For examples,
\begin{eqnarray}
H_{1}^{(2)}&=&F_2=\sum_{i=1}^{N}\frac{\partial^2}{\partial z^2_i}z_i
-2\sum_{i=1}^{N}\sum_{j\neq i}\frac{\partial}{\partial z_i}\frac{z_i}{z_i-z_j}-F_1,\nonumber\\
H_{1}^{(3)}&=&F_3=\sum_{i=1}^N\frac{\partial^3 }{\partial z_i^3}z_i^{2}
-3\sum_{i=1}^N\sum_{j\neq i}\frac{\partial^2 }{\partial z_i^2}\frac{z_i^{2}}{z_i-z_j}+3\sum_{i=1}^N\sum_{k\neq j\neq i}
\frac{\partial }{\partial z_i}\frac{z_i^{2}}{(z_i-z_j)(z_i-z_k)}\nonumber\\
&&-3F_2-2F_1,\nonumber\\
H_{2}^{(1)}&=&\sum_{i=1}^{N}\frac{\partial^2}{\partial z^2_i}
-2\sum_{i=1}^{N}\sum_{j\neq i}\frac{\partial}{\partial z_i}\frac{1}{z_i-z_j},\nonumber\\
H_{2}^{(2)}&=&\sum_{i=1}^N\frac{\partial^4 }{\partial z_i^4}z_i^{2}-4\sum_{i=1}^N\sum_{j\neq i}
\frac{\partial^3 }{\partial z_i^3}\frac{z_i^{2}}{z_i-z_j}
+3\sum_{i=1}^N\sum_{k\neq j\neq i}\frac{\partial^2 }{\partial z^2_i}\frac{z_i^{2}}{(z_i-z_j)(z_i-z_k)}
\nonumber\\
&&-12\sum_{i=1}^N\sum_{l\neq k\neq j\neq i}\frac{\partial}{\partial z_i}\frac{z^2_i}{(z_i-z_j)
(z_i-z_k)(z_i-z_l)}-2\sum_{i=1}^N\frac{\partial^3}{\partial z^3_i}z_i\nonumber\\
&&-6\sum_{i=1}^N\sum_{k\neq j\neq i}\frac{\partial }{\partial z_i}\frac{z_i}{(z_i-z_j)(z_i-z_k)}
-12\sum_{i=1}^N\sum_{k\neq j\neq i}\frac{\partial^2 }{\partial z_i\partial z_j}\frac{z^2_iz_j}{(z_i-z_j)^2(z_i-z_k)}\nonumber\\
&&-6\sum_{i=1}^N\sum_{l\neq k\neq j\neq i}\frac{\partial }{\partial z_i}
\frac{z_iz_l(z_jz_k-z_iz_l)}{(z_i-z_j)(z_i-z_k)(z_l-z_j)(z_l-z_k)(z_i-z_l)}\nonumber\\
&&+6\sum_{i=1}^N\sum_{k\neq j\neq i}\frac{\partial^2 }{\partial z^2_i}
\frac{z^2_iz_j(z_i+z_j-2z_k)}{(z_i-z_j)^2(z_i-z_k)(z_j-z_k)}+2H_{2}^{(1)}.
\end{eqnarray}

By inserting the operators (\ref{Hnm}) into the integrand of (\ref{PF}), we obtain the constraints
\begin{equation}\label{Hnmcons}
{\hat H}_{-n}^{(m)}Z=\int d^Nz\ \ H_n^{(m)}\Delta(z)^2 e^{\sum_{i=1}^{N}V(z_i)}=0,
\ \ \ n,m \in \mathbb{Z}_{+},
\end{equation}
where
\begin{equation}\label{adhatw-rr}
{\hat H}_{-n}^{(m)}=\frac{1}{(n-1)!}{\rm ad}^{n-1}_{{\hat F}_{m+1}}{\hat F}_{m},
\end{equation}
${\hat F}_{m}={\rm ad}^{m-1}_{\frac{1}{2}{\hat W}_0}{{\hat F}}_{1}$, the operators
${\hat F}_1$ and ${\hat W}_0$ are respectively given by
\begin{eqnarray}\label{F1W0}
{\hat F}_1&=&{\hat L}_{-1}=\sum_{k=1}^{\infty}kt_k\frac{\partial}{\partial t_{k-1}},\nonumber\\
{\hat W}_0&=&\sum_{k=1}^{\infty}kt_k\hat L_{k}\nonumber\\
&=&\sum_{k,l=1}^{\infty}klt_kt_l\frac{\partial}{\partial t_{k+l}}
+\sum_{k=1}^{\infty}\sum_{l=0}^{k}kt_k\frac{\partial}{\partial t_l}\frac{\partial}{\partial t_{k-l}}.
\end{eqnarray}

From (\ref{adhatw-rr}) and (\ref{F1W0}), we find that the constraint operators ${\hat H}_{-n}^{(m)}$ are reducible to the Virasoro constraint operators (\ref{virconop}).

For examples,
\begin{eqnarray}
{\hat H}_{-1}^{(2)}&=&{\hat F}_2=\sum_{k=1}^{\infty}kt_k\hat L_{k-1},\nonumber\\
{\hat H}_{-1}^{(3)}&=&{\hat F}_3=\sum_{k,l=0}^{\infty}kt_k\frac{\partial}{\partial t_{l}}\hat L_{k-l-1}
+\sum_{k,l=1}^{\infty}klt_kt_l\hat L_{k+l-1},\nonumber\\
{\hat H}_{-2}^{(1)}&=&\sum_{k=1}^{\infty}kt_k\hat L_{k-2},\nonumber\\
{\hat H}_{-2}^{(2)}&=&\sum_{k,l=1}^{\infty}klt_k\hat L_{k-l-1}\hat L_{l-1}
-2\sum_{k,l,p=0}^{\infty}klt_k\frac{\partial}{\partial t_{l}}\frac{\partial}{\partial t_{p}}\hat L_{k-l-p-2}\nonumber\\
&&+\sum_{k,l=1}^{\infty}kl(2l-k+1)t_k\hat L_{k-2}+\sum_{k,l,p=0}^{\infty}kl(p-k+3)t_kt_l
\frac{\partial}{\partial t_{p}}\hat L_{k+l-p-2}\nonumber\\
&&+\sum_{k,l,p=1}^{\infty}klpt_kt_lt_p\hat L_{k+l+p-2}.
\end{eqnarray}

Let us take the rescaling variables $p_k=kt_k$ $(k>0)$ and substitute
$\frac{\partial}{\partial t_0}$ by $N$ in (\ref{adhatw-rr}), then they become
\begin{equation}\label{padhatw-rr}
{\hat H}_{-n}^{(m)}\{p\}=\frac{1}{(n-1)!}{\rm ad}^{n-1}_{{\hat F}_{m+1}\{p\}}{\hat F}_{m}\{p\},
\end{equation}
where ${\hat F}_{m}\{p\}={\rm ad}^{m-1}_{\frac{1}{2}{\hat W}_0\{p\}}{{\hat F}}_{1}\{p\}$, the operators
${\hat F}_1\{p\}$ and ${\hat W}_0\{p\}$ are respectively given by
\begin{eqnarray}\label{pF1W0}
{\hat F}_1\{p\}&=&\sum_{k=1}^{\infty}kp_{k+1}\frac{\partial}{\partial p_{k}}+Np_1,\nonumber\\
{\hat W}_0\{p\}&=&\sum_{k,l=1}^{\infty}\big((k+l)p_{k}p_l
\frac{\partial}{\partial p_{k+l}}+klp_{k+l}\frac{\partial}{\partial p_k}\frac{\partial}{\partial p_l}\big)+2N\sum_{k=1}^{\infty}kp_{k}\frac{\partial}{\partial p_{k}}.
\end{eqnarray}

The intriguing result is that the operators ${\hat H}_{-n}^{(m)}\{p\}$
can be used to generate the negative branch of the partition functions \cite{wangliu,2301.04107,1405.1395}
\begin{eqnarray}\label{z-m}
Z_{-}^{(m)}
&=&{\rm exp}(\sum_{k=1}^{\infty}\frac{g_k{\hat H}_{-k}^{(m)}\{p\}}{k})\cdot 1\nonumber\\
&=&\sum_{\lambda}\left(\frac{S_{\lambda}\{p_k=N\}}{S_{\lambda}\{p_k=\delta_{k,1}\}}\right)^m
S_{\lambda}\{p\}S_{\lambda}\{g\}\nonumber\\
&=&\prod_{l=1}^m\int\int_{N\times N} dX_ldY_l {\rm exp}(-\sum_{j=1}^m{\rm Tr}X_jY_j+\sum_{k=1}^{\infty}\frac{1}{k}\sum_{j=1}^{m-1}
{\rm Tr}Y_j^k{\rm Tr}X_{j+1}^k)\nonumber\\
&&\times {\rm exp}(\sum_{k=1}^{\infty}\frac{1}{k}(p_k{\rm Tr}X_1^k+g_k{\rm Tr}Y_m^k)),
\end{eqnarray}
where $S_{\lambda}$ is the Schur function associated with the partition $\lambda$,
$\frac{S_{\lambda}\{p_k=N\}}{S_{\lambda}\{p_k=\delta_{k,1}\}}=\prod_{(i,j)\in \lambda}(N+j-i)$, $X_l$ and $Y_l$, $l=1,2,\cdots,m$, are Hermitian and anti-Hermitian $N\times N$ matrices, respectively.
When particularize to $m=1$, $g_k=\delta_{k,2}$ and $m=2$, $g_k=\delta_{k,1}$, (\ref{z-m}) reduce to the Gaussian Hermitian \cite{Shakirov2009,AMironov1705}
and $N\times N$ complex matrix models \cite{AlexandrovJHEP2009,AMironov1705}, respectively.

Let us construct the $r$-th total derivative operators
\begin{eqnarray}\label{deroperators}
W_{n}^{r}&:=&\sum_{k=0}^{r-1}(-1)^kC_{r}^k
\sum_{i=1}^N
\frac{\partial^{r-k} }{\partial z_i^{r-k}}z_i^{n+r}\Delta^{-1}
\frac{\partial^{k} \Delta}{\partial z_i^{k}}\nonumber\\
&=&\sum_{k=0}^{r-1}(-1)^kC_{r}^k
\sum_{i=1}^N\sum_{j_1\neq \cdots \neq j_k\neq i}
\frac{\partial^{r-k} }{\partial z_i^{r-k}}
\frac{z_i^{n+r}}{(z_i-z_{j_1})\cdots (z_i-z_{j_k})},
\end{eqnarray}
where $r\geq 1$, $n\geq -r$ and $C_{r}^k=\frac{r!}{k!(r-k)!}$.

The first few operators are
\begin{eqnarray}
W_{n}^{1}&=&\sum_{i=1}^{N}\frac{\partial}{\partial z_i}z_i^{n+1},\nonumber\\
W^2_{n}&=&\sum_{i=1}^{N}\frac{\partial^2}{\partial z^2_i}z^{n+2}_i
-2\sum_{i=1}^{N}\sum_{j\neq i}\frac{\partial}{\partial z_i}\frac{z^{n+2}_i}{z_i-z_j},\nonumber\\
W_{n}^{3}&=&\sum_{i=1}^N\frac{\partial^3 }{\partial z_i^3}z_i^{n+3}
-3\sum_{i=1}^N\sum_{j\neq i}\frac{\partial^2 }{\partial z_i^2}\frac{z_i^{n+3}}{z_i-z_j}+3\sum_{i=1}^N\sum_{k\neq j\neq i}\frac{\partial }
{\partial z_i}\frac{z_i^{n+3}}{(z_i-z_j)(z_i-z_k)}.
\end{eqnarray}
We see that
\begin{eqnarray}
F_2&=&W_{-1}^{2}-W_{-1}^{1},\nonumber\\
F_3&=&W_{-1}^{3}-3W_{-1}^{2}+W_{-1}^{1},\nonumber\\
H_{r}^{(1)}&=&W_{-r}^{r}=\sum_{k=0}^{r-1}(-1)^kC_{r}^k
\sum_{i=1}^N\sum_{j_1\neq \cdots \neq j_k\neq i}
\frac{\partial^{r-k} }{\partial z_i^{r-k}}
\frac{1}{(z_i-z_{j_1})\cdots (z_i-z_{j_k})}.\label{Hr1}
\end{eqnarray}

Notice that under the transformation $(-1)^{k}\frac{\partial^{r-k} }{\partial z^{r-k}}f(z)\rightarrow f(z)\frac{\partial^{r-k} }{\partial z^{r-k}}$,
the operators (\ref{Hr1}) become
\begin{eqnarray}\label{hh}
{\bar H}_{r}^{(1)}&=&\sum_{k=0}^{r-1}C_{r}^k
\sum_{i=1}^N\sum_{j_1\neq \cdots \neq j_k\neq i}
\frac{1}{(z_i-z_{j_1})\cdots (z_i-z_{j_k})}\frac{\partial^{r-k} }{\partial z_i^{r-k}}\nonumber\\
&=&\sum_{i=1}^{N}D^r_i,
\end{eqnarray}
where the operators $D_i$ are given by
$D_i=\frac{\partial }{\partial z_i}+\sum_{j\neq i}\frac{1}{z_i-z_j}$ \cite{2306.06623}. The operators
$\Delta \circ {\bar H}_r^{(1)}\circ \Delta^{-1}$ describe free particles, which are associated
with the so called free fermion point of the Calogero system \cite{2306.06623,2303.05273}.

There is a similar expression for $W_{-r}^{r}$, i.e.,
\begin{equation}
W_{-r}^{r}=\sum_{i=1}^{N}{\tilde D}^r_i,
\end{equation}
where $\tilde D_i=\frac{\partial }{\partial z_i}-\sum_{j\neq i}\frac{1}{z_i-z_j}$.

In general, the operators (\ref{deroperators}) can be expressed as
\begin{equation}\label{wnr1}
W_{n}^{r}=\sum_{i=1}^{N}{\tilde D}^r_iz^{n+r}_i-(-1)^r\sum_{i=1}^N\sum_{j_1\neq \cdots \neq j_r\neq i}
\frac{z^{n+r}_i}{(z_i-z_{j_1})\cdots (z_i-z_{j_r})}.
\end{equation}

For the convenience of later discussions, we give the commutators
\begin{eqnarray}\label{wwcomm}
[W^1_n, W^r_m]&=&(m-rn)W_{m+n}^r+\sum_{k=2}^{r}(-1)^{k}
C^{k}_rA^{k}_{n+1}(1+\frac{1}{r+1-k})W_{m+n}^{r+1-k}\nonumber\\
&&+\sum_{l=2}^{r-1}(-1)^{l}\sum_{k_1=1}^{n-1}\sum_{k_2=0}^{k_1-1}\cdots\sum_{k_l=0}^{k_{l-1}-1}
(n-k_1)A^{l}_{r}W_{m+n-k_l}^{r-l}W_{k_l}^0\nonumber\\
&&-r\sum_{k=0}^{n-1}(n-k)W_{m+n-k}^{r-1}W_{k}^0,
\end{eqnarray}
where $A_n^k=n(n-1)\cdots (n-k+1)$ and $W_{n}^0=\sum_{i=1}^Nz_i^n$.

Let us list several commutators of (\ref{wwcomm})
\begin{align}
[W^1_n, W^2_m]=&(m-2n)W^2_{m+n}+2n(n+1)W^1_{m+n}-2\sum_{k=0}^{n-1}(n-k)W_{m+n-k}^1W_k^0,
\nonumber\\
[W^1_n, W^3_m]=&-2(n+1)n(n-1)W^1_{m+n}+\frac{9}{2}n(n+1)W^2_{m+n}+(m-3n)W^3_{m+n}\nonumber\\
&-3\sum_{k=0}^{n-1}(n-k)W_{m+n-k}^2W_k^0+6\sum_{l=0}^{n-1}\sum_{k=0}^{l-1}(n-l)
W_{m+n-k}^1W_k^0,\nonumber\\
[W^1_n, W^4_m]=&2(n+1)n(n-1)(n-2)W^1_{m+n}-6(n+1)n(n-1)W^2_{m+n}+8(n+1)nW^3_{m+n}\nonumber\\
&+(m-4n)W^4_{m+n}-4\sum_{k=0}^{n-1}(n-k)W_{m+n-k}^3W_k^0
+12\sum_{k_1=0}^{n-1}\sum_{k_2=0}^{k_1-1}
(n-k_1)W_{m+n-k_2}^2W_{k_2}^0\nonumber\\
&-24\sum_{k_1=0}^{n-1}\sum_{k_2=0}^{k_1-1}\sum_{k_3=0}^{k_2-1}
(n-k_1)W_{m+n-k_3}^1W_{k_3}^0.
\end{align}

By inserting (\ref{deroperators}) into the partition function (\ref{PF}), we obtain a series of constraints
\begin{equation}\label{wcons1}
\hat W_n^r Z=0,\ \ n\geq -r,
\end{equation}
where the constraint operators $\hat W_n^r$ can be written out explicitly by
using the formula \cite{Itoyama}
\begin{equation}
\Delta^{-1}\sum_{i=1}^{N}\frac{1}{p-z_i}
\frac{\partial^{n} \Delta}{\partial z_i^{n}}=\frac{1}{n+1}
(\frac{\partial }{\partial p}+\sum_{i=1}^{N}\frac{1}{p-z_i})^{n+1}\cdot 1,
\quad n<N.
\end{equation}

For examples,
\begin{eqnarray}
\hat W_n^2&=&\sum_{k=1}^{\infty}kt_k\hat L_{n+k}+(n+2)\hat L_n,\label{hwn2}\\
\hat W_n^3&=&\sum_{k=1}^{\infty}kt_k\hat W^2_{n+k}+\sum_{k=1}^{\infty}k(n+k+2)t_k
\hat L_{n+k}+\frac{3}{2}(n+3)\hat W^2_{n}+\frac{1}{2}(5n^2+24n+29)\hat L_n\nonumber\\
&&+\frac{1}{2}\sum_{k=0}^{\infty}k(k^2+kn)t_k\frac{\partial }{\partial t_{n+k}}
+\frac{1}{2}\sum_{k=0}^{n}(k^2+15n+5n^2-30k-11kn)\frac{\partial }{\partial t_{k}}\frac{\partial }{\partial t_{n-k}}\nonumber\\
&&+\frac{3}{2}\sum_{k=0}^{\infty}\sum_{l=0}^{n+k}k(n+k-2l)t_k\frac{\partial }{\partial t_{l}}\frac{\partial }{\partial t_{n+k-l}}
+\sum_{k=0}^{n}\sum_{l=0}^{n-k}(-2n+3k+3l)\frac{\partial }{\partial t_{k}}\frac{\partial }{\partial t_{l}}\frac{\partial }{\partial t_{n-k-l}}
\nonumber\\
&&+\frac{1}{2}\sum_{k=0}^{n}\sum_{l=0}^{n-k}\sum_{p=0}^{n-k-l}\frac{\partial }{\partial t_{k}}\frac{\partial }{\partial t_{l}}
\frac{\partial }{\partial t_{p}}\frac{\partial }{\partial t_{n-k-l-p}}
+\sum_{k=0}^{\infty}\sum_{l=0}^{n+k}\sum_{p=0}^{n+k-l}kt_k\frac{\partial }{\partial t_{l}}\frac{\partial }{\partial t_{p}}
\frac{\partial }{\partial t_{n+k-l-p}}\nonumber\\
&&+\frac{1}{2}\sum_{k=0}^{\infty}\sum_{l=0}^{\infty}\sum_{p=0}^{n+k+l}
klt_kt_l\frac{\partial }{\partial t_{p}}\frac{\partial }{\partial t_{n+k+l-p}}.
\end{eqnarray}

\begin{theorem}\label{thm}
For $r\geq 2$, the constraint operators $\hat W_n^r$ ($n\geq -r$) are reducible to the Virasoro constraint operators (\ref{virconop}).
\end{theorem}
\begin{proof}
Let us prove the theorem by mathematical induction.
From (\ref{hwn2}), we see that the theorem holds for $r=2$.
Let us assume that the theorem holds for $r-1$, i.e, $\hat W_n^{r-1}$ ($n\geq -r+1$) are reducible to the Virasoro constraint operators (\ref{virconop}).

Taking $n=1$ in (\ref{wwcomm}), we have
\begin{equation}\label{repnwrm}
W^r_{m+1}=\frac{1}{m-r}[W^1_{1}, W^r_m]-\frac{r(r-N)}{m-r}W^{r-1}_{m+1}, \ \ m\neq r.
\end{equation}
Then by inserting (\ref{repnwrm}) into the partition function (\ref{PF}), it is easy to give
\begin{eqnarray}\label{thmeqn1}
{\hat W}^r_{-r+1}&=&\frac{1}{2r}[{\hat W}^1_{1}, {\hat W}^r_{-r}]+\frac{r-N}{2}{\hat W}^{r-1}_{-r+1},\nonumber\\
{\hat W}^r_{-r+2}&=&\frac{1}{2r-1}[{\hat W}^1_{1}, {\hat W}^r_{-r+1}]+\frac{r(r-N)}{2r-1}{\hat W}^{r-1}_{-r+2},\nonumber\\
&\vdots&\nonumber\\
{\hat W}^r_{r-1}&=&\frac{1}{2}[{\hat W}^1_{1}, {\hat W}^r_{r-2}]+\frac{r(r-N)}{2}{\hat W}^{r-1}_{r-1},\nonumber\\
{\hat W}^r_{r}&=&[{\hat W}^1_{1}, {\hat W}^r_{r-1}]+r(r-N){\hat W}^{r-1}_{r}.
\end{eqnarray}
Note that ${\hat W}_{-r}^{r}={\hat H}_{-r}^{(1)}$, they can be reduced to the Virasoro constraint operators (\ref{virconop}).
Then by the inductive assumption and relations (\ref{thmeqn1}), we see the theorem holds for $\hat W_n^r$ ($-r\leq n\leq r$).

Using the commutator $[{W}^{1}_{2}, {W}^{r}_{r-1}]$, we obtain
\begin{equation}
{\hat W}^r_{r+1}=\frac{1}{r+1}\left(r(3r-2N){\hat W}^{r-1}_{r+1}+r(r-1)(N-r+1)
{\hat W}^{r-2}_{r+1}-r{\hat W}^{0}_{1}{\hat W}^{r-1}_{r}+[{\hat W}^{1}_{2}, {\hat W}^{r}_{r-1}]
\right),
\end{equation}
where ${\hat W}^{0}_{1}=\frac{\partial}{\partial t_1}$.
It shows that the theorem holds for ${\hat W}^r_{r+1}$.

Then using the relation (\ref{repnwrm}) for $m\geq r+1$ and inductive assumption, we also confirm that ${\hat W}^{r}_{n}$ ($n>r+1$)
are reducible to the Virasoro constraint operators (\ref{virconop}).
Therefore, the theorem holds for $\hat W_n^r$ ($n\geq -r$).
\end{proof}

Let us turn to the constraints (\ref{dnm}).
The direct calculations show that (\ref{dnm}) can be expressed as
\begin{equation}\label{dnm1}
\tilde W_n^rZ=\int d^Nz\ \ \sum_{l=0}^{r-1}(-1)^lC_r^lA_{n+r}^lW_n^{r-l}(\Delta^2e^V)=0.
\end{equation}
From (\ref{dnm1}), we reach the desired result for the constraint operators $\tilde W_n^r$
\begin{equation}
\tilde W_n^r=\sum_{l=0}^{r-1}(-1)^lC_r^lA_{n+r}^l{\hat W}_n^{r-l}.
\end{equation}
It indicates that the Itoyama-Matsuo conjecture is valid by the Theorem \ref{thm}.

In order to further understand the constraint operators, let us rewrite the total derivatives (\ref{deroperators}) as
\begin{eqnarray}
{W}_{n}^{r}&=&\mathbb{W}_{n}^{r}-\sum_{k=1}^{r-1}(-1)^k(C_{r-1}^k-C_{r-1}^{k-1})
\sum_{i=1}^N
\frac{\partial^{r-k} }{\partial z_i^{r-k}}z_i^{n+r}\Delta^{-1}
\frac{\partial^{k} \Delta}{\partial z_i^{k}}\nonumber\\
&&-2\sum_{k=2}^{r-1}\sum_{l=1}^{\lfloor \frac{k}{2}\rfloor}(-1)^kC_{r-1}^kC_{k}^l\sum_{i=1}^N
\frac{\partial^{r-k} }{\partial z_i^{r-k}}z_i^{n+r}\Delta^{-2}
\frac{\partial^{k-l} \Delta}{\partial z_i^{k-l}}\frac{\partial^{l} \Delta}{\partial z_i^{l}}\nonumber\\
&&+\sum_{k=1}^{\lfloor \frac{r-1}{2}\rfloor}C_{r-1}^{2k}C_{2k}^{k}\sum_{i=1}^N
\frac{\partial^{r-2k} }{\partial z_i^{r-2k}}z_i^{n+r}\left(\Delta^{-1}
\frac{\partial^{k} \Delta}{\partial z_i^{k}}\right)^2,
\end{eqnarray}
where the higher order total derivative operators $\mathbb{W}_{n}^{r}$ are given by
\begin{eqnarray}\label{anwnr}
\mathbb{W}_{n}^{r}=\sum_{k=0}^{r-1}(-1)^kC_{r-1}^k
\sum_{i=1}^N\frac{\partial^{r-k} }{\partial z_i^{r-k}}z_i^{n+r}
\Delta^{-2}
\frac{\partial^{k} \Delta^2}{\partial z_i^{k}}, \ \ r\geq 1, n\geq -r.
\end{eqnarray}

Then by inserting (\ref{anwnr}) into the partition function (\ref{PF}) and using
\begin{eqnarray}
\mathbb{W}_{n}^{r}(\Delta^{2}e^V)
&=&\sum_{k=0}^{r-1}C_{r-1}^kA_{n+r}^k\left(2\sum_{j\neq i}\frac{z_i^{n+r-k}Q_{r-k-1}[\partial_{z_i}V]}{z_i-z_j}
+z_i^{n+r-k}Q_{r-k}[\partial_{z_i}V]
\right.\nonumber\\
&&\left.+(n+r-k)z_i^{n+r-k-1}Q_{r-k-1}[\partial_{z_i}V]\right)\Delta^{2}e^V,
\end{eqnarray}
where $Q_s[f]=(\frac{\partial}{\partial P}+f(P))^s\cdot 1$, we obtain the constraints
\begin{equation}\label{wcons2}
{\mathbb{\hat W}}_{n}^{r}Z=0, \ \ n\geq -r,
\end{equation}
where
\begin{equation}\label{wconsop2}
{\mathbb{\hat W}}_{n}^{r}=\sum_{k=0}^{r-1}C_{r-1}^kA_{n+r}^k\rho(P^{n+r-k}Q_{r-k-1}[j(P)]),
\end{equation}
$j(P)=\sum_{k=1}^{\infty}kt_kP^{k-1}$, $\rho(P^{n+1})={\hat L}_n$.

For examples,
\begin{eqnarray}
\mathbb{\hat W}_n^3&=&\sum_{k,l=1}^{\infty}klt_kt_l{\hat L}_{n+k+l}+2A_{n+3}^1\sum_{k=1}^{\infty}kt_k{\hat L}_{n+k}
+A_{n+3}^2{\hat L}_{n}+\sum_{k=2}^{\infty}k(k-1)t_k\hat L_{n+k},\nonumber\\
\mathbb{\hat W}_n^4&=&A_{n+4}^3{\hat L}_{n}+3A^2_{n+4}\sum_{k=1}^{\infty}kt_k{\hat L}_{n+k}
+3A^1_{n+4}\sum_{k=1}^{\infty}klt_kt_l{\hat L}_{n+k+l}\nonumber\\
&&+3A^1_{n+4}\sum_{k=2}^{\infty}k(k-1)t_k\hat L_{n+k}+\sum_{k_1,k_2,k_3=1}^{\infty}k_1k_2k_3t_{k_1}t_{k_2}t_{k_3}{\hat L}_{n+k_1+k_2+k_3}\nonumber\\
&&+\sum_{k=3}^{\infty}k(k-1)(k-2)t_k\hat L_{n+k}+3\sum_{k,l=1}^{\infty}kl(l-1)t_kt_l{\hat L}_{n+k+l}.
\end{eqnarray}

From (\ref{wconsop2}), we see that the constraint operators ${\mathbb{\hat W}}_{n}^{r}$ are explicitly represented by the Virasoro constraint operators (\ref{virconop}).
However, it is not easy to give the expressions of constraint operators ${\hat W}_n^{r}$ in terms of the Virasoro constraint operators (\ref{virconop}).

\section{Higher order constraints for the $\beta$-deformed Hermitian matrix model}
The $\beta$-deformed matrix model is given by
\begin{equation}\label{DPF}
Z_{\beta}
=\int d^Nz\ \ \Delta(z)^{2\beta} e^{\sum_{i=1}^{N}V(z_i)},
\end{equation}
which satisfies the Virasoro constraints \cite{Awata94}
$$\mathcal {\hat L}_{n}Z_{\beta}=0,\ \ \ n\geq -1,$$
where
\begin{eqnarray}\label{defvirasoro}
\mathcal {\hat L}_{n}=\sum_{k=1}^{\infty}kt_k\frac{\partial}{\partial t_{k+n}}+
\beta\sum_{k=0}^{n}\frac{\partial}{\partial t_{k}}
\frac{\partial}{\partial t_{n-k}}+(1-\beta)(n+1)\frac{\partial}{\partial t_{n}}.
\end{eqnarray}

It was shown that by inserting the $\beta$-deformed total derivative operators \cite{Wangr2022}
\begin{equation}\label{barMn}
\bar {\mathcal{W}}_n=\sum_{i=1}^{N}\frac{\partial^2}{\partial z^2_i}z^{n+2}_i
-2\beta\sum_{i=1}^{N}\sum_{j\neq i}\frac{\partial}{\partial z_i}\frac{z^{n+2}_i}{z_i-z_j}-(n+2)
\sum_{i=1}^{N}\frac{\partial}{\partial z_i}z^{n+1}_i, \ \ n\geq -2,
\end{equation}
into (\ref{DPF}), there are the constraints
\begin{equation}
\hat {\mathcal{W}}_nZ_{\beta}=0,\ \ n\geq -2,
\end{equation}
where
\begin{eqnarray}\label{DMnt}
\hat {\mathcal{W}}_{n}&=&\sum_{k=1}^{\infty}kt_k\mathcal {\hat L}_{n+k}\nonumber\\
&=&\sum_{k,l=1}^{\infty}klt_kt_l\frac{\partial}{\partial t_{k+l+n}}
+\beta\sum_{k=1}^{\infty}\sum_{l=0}^{k+n}kt_k\frac{\partial}{\partial t_l}\frac{\partial}{\partial t_{k-l+n}}\nonumber\\
&&+(1-\beta)\sum_{k=0}^{\infty}kt_k(n+k+1)\frac{\partial}{\partial t_{k+n}}.
\end{eqnarray}

In similarity with the operator $\bar W_0$ (\ref{Mnz0}) case,
for the operator $\bar {\mathcal{W}}_0$ in (\ref{barMn})
\begin{equation}
\bar {\mathcal{W}}_0=\sum_{i=1}^{N}\frac{\partial^2}{\partial z^2_i}z^{2}_i
-2\beta\sum_{i=1}^{N}\sum_{j\neq i}\frac{\partial}{\partial z_i}\frac{z^{2}_i}{z_i-z_j}-2
\sum_{i=1}^{N}\frac{\partial}{\partial z_i}z_i,
\end{equation}
we may construct the commutative operators
\begin{equation}\label{dHnm}
\mathcal{H}_n^{(m)}=\frac{(-1)^{n-1}}{(n-1)!}{\rm ad}^{n-1}_{{\mathcal F}_{m+1}}{\mathcal F}_{m},\ \ \ n,m \in \mathbb{Z}_{+},
\end{equation}
where ${\mathcal F}_{m}={\rm ad}^{m-1}_{-\frac{1}{2}\bar {\mathcal{W}}_0}{\mathcal{F}}_{1}$, $\mathcal{F}_1=\sum_{i=1}^{N}\frac{\partial}{\partial z_i}$.
When $\beta=1$, the operators (\ref{dHnm}) reduce to (\ref{Hnm}).

For examples,
\begin{eqnarray}
{\mathcal H}_{1}^{(2)}&=&{\mathcal F}_2=\sum_{i=1}^{N}\frac{\partial^2}{\partial z^2_i}z_i
-2\beta\sum_{i=1}^{N}\sum_{j\neq i}\frac{\partial}{\partial z_i}\frac{z_i}{z_i-z_j}-{\mathcal F}_1,\nonumber\\
{\mathcal H}_{1}^{(3)}&=&{\mathcal F}_3=\sum_{i=1}^N\frac{\partial^3 }{\partial z_i^3}z_i^{2}
-3\beta\sum_{i=1}^N\sum_{j\neq i}\frac{\partial^2 }{\partial z_i^2}\frac{z_i^{2}}{z_i-z_j}\nonumber\\
&&+3\beta^2\sum_{i=1}^N\sum_{k\neq j\neq i}
\frac{\partial }{\partial z_i}\frac{z_i^{2}}{(z_i-z_j)(z_i-z_k)}\nonumber\\
&&+2\beta(\beta-1)\sum_{i=1}^{N}\sum_{j\neq i}\frac{\partial}{\partial z_i}\frac{z_i}{z_i-z_j}-3\mathcal{F}_2-2\mathcal{F}_1,\nonumber\\
{\mathcal H}_{2}^{(1)}&=&\sum_{i=1}^{N}\frac{\partial^2}{\partial z^2_i}
-2\beta\sum_{i=1}^{N}\sum_{j\neq i}\frac{\partial}{\partial z_i}\frac{1}{z_i-z_j}.
\end{eqnarray}

The operators (\ref{dHnm}) give rise to the constraints for the $\beta$-deformed matrix model (\ref{DPF})
\begin{equation}\label{dHnmcons}
\mathcal{\hat H}_{-n}^{(m)}Z_{\beta}=0,\ \ \ n,m \in \mathbb{Z}_{+},
\end{equation}
where
\begin{equation}\label{dhatHnm}
\mathcal{\hat H}_{-n}^{(m)}=\frac{1}{(n-1)!}{\rm ad}^{n-1}_{{\mathcal {\hat F}}_{m+1}}{\mathcal {\hat F}}_{m},
\end{equation}
${\mathcal {\hat F}}_{m}={\rm ad}^{m-1}_{\frac{1}{2}\mathcal {\hat W}_0}{\mathcal {\hat F}}_{1}$, $\mathcal {\hat F}_{1}=\mathcal {\hat L}_{-1}=\sum_{k=1}^{\infty}kt_k\frac{\partial}{\partial t_{k-1}}$ and
$\hat {\mathcal{W}}_{0}=\sum_{k=1}^{\infty}kt_k\mathcal {\hat L}_{k}$.

Similar to the constraint operators (\ref{adhatw-rr}) of the Hermitian matrix model, the constraint operators (\ref{dhatHnm}) are reducible to the Virasoro constraint operators (\ref{defvirasoro}).

For examples,
\begin{eqnarray}
\mathcal{\hat H}_{-1}^{(2)}&=&\mathcal{\hat F}_2=\sum_{k=1}^{\infty}kt_k\mathcal {\hat L}_{k-1},\nonumber\\
\mathcal{\hat H}_{-1}^{(3)}&=&\mathcal{\hat F}_3=
\beta\sum_{k,l=0}^{\infty}kt_k\frac{\partial}{\partial t_{l}}\mathcal {\hat L}_{k-l-1}
+\sum_{k,l=1}^{\infty}klt_kt_l\mathcal {\hat L}_{k+l-1}\nonumber\\
&&+(1-\beta)\sum_{k=0}^{\infty}(k+1)^2t_{k+1}\mathcal {\hat L}_{k}\nonumber,\\
\mathcal{\hat H}_{-2}^{(1)}&=&\sum_{k=1}^{\infty}kt_k\mathcal {\hat L}_{k-2}.
\end{eqnarray}

When we take the rescaling variables $p_k=\beta^{-1}kt_k$ $(k>0)$, and substitute
$\frac{\partial}{\partial t_0}$ by $N$  in (\ref{dhatHnm}), they become
\begin{equation}\label{pdhatHnm}
\mathcal{\hat H}_{-n}^{(m)}\{p\}=\frac{1}{(n-1)!}{\rm ad}^{n-1}_{{\mathcal {\hat F}}_{m+1}\{p\}}{\mathcal {\hat F}}_{m}\{p\},
\end{equation}
where ${\mathcal {\hat F}}_{m}\{p\}={\rm ad}^{m-1}_{\frac{1}{2}\mathcal {\hat W}_0\{p\}}{\mathcal {\hat F}}_{1}\{p\}$, the operators ${\mathcal {\hat F}}_1\{p\}$ and $\mathcal {\hat W}_0\{p\}$ are respectively given by
\begin{eqnarray}\label{pdF1W0}
\mathcal {\hat F}_1\{p\}&=&\sum_{k=1}^{\infty}kp_{k+1}\frac{\partial}{\partial p_{k}}
+\beta Np_1,\nonumber\\
\hat {\mathcal{W}}_0\{p\}&=&\sum_{k,l=1}^{\infty}\big(\beta(k+l)p_{k}p_l
\frac{\partial}{\partial p_{k+l}}+klp_{k+l}\frac{\partial}{\partial p_k}\frac{\partial}{\partial p_l}\big)\nonumber\\
&&+2\beta N\sum_{k=1}^{\infty}kp_{k}
\frac{\partial}{\partial p_{k}}+(1-\beta)\sum_{k=1}^{\infty}(k+1)kp_k\frac{\partial}{\partial p_{k}}.
\end{eqnarray}

The operators (\ref{pdhatHnm}) coincide with the $W$-operators
${\mathcal{\hat W}}_{-n}^{(m)}(\vec{u})$ with $u_1=N$, $u_2=u_3=\cdots=u_m=N+\beta^{-1}-1$ constructed in Ref.\cite{2301.12763}, which generate the $\beta$-deformed partition functions
\begin{eqnarray}\label{dz-m}
\mathcal{Z}_{-}^{(m)}
&=&{\rm exp}(\beta\sum_{k=1}^{\infty}\frac{g_k\mathcal{\hat H}_{-k}^{(m)}\{p\}}{k})\cdot 1\nonumber\\
&=&\sum_{\lambda}\frac{J_{\lambda}\{p_k=N\}(J_{\lambda}\{p_k=N+\beta^{-1}-1\})^{m-1}}
{(J_{\lambda}\{p_k=\beta^{-1}\delta_{k,1}\})^m}
\frac{J_{\lambda}\{p\}J_{\lambda}\{g\}}{\langle J_{\lambda},J_{\lambda}\rangle},
\end{eqnarray}
where $J_{\lambda}$ is the Jack function associated with the partition $\lambda$,
$\frac{J_{\lambda}\{p_k=u\}}{J_{\lambda}\{p_k=\beta^{-1}\delta_{k,1}\}}=\prod_{(i,j)\in \lambda}(j-1+\beta(u-i+1))$, $\langle J_{\lambda},J_{\lambda}\rangle=\prod_{(i,j)\in \lambda}\frac{\lambda_i-j+1+\beta({\lambda}^{'}_j-i)}{\lambda_i-j+\beta({\lambda}^{'}_j-i+1)}$ and
$\lambda^{'}=(\lambda_1^{'},\lambda_2^{'},\cdots)$ is the conjugate partition of $\lambda$.
The $\beta$-deformed Gaussian Hermitian \cite{Morozov1901} and $N\times N$ complex matrix models \cite{Cheny,Cassia2020} are contained in (\ref{dz-m}).

For the operators
\begin{equation}\label{Hn1}
\mathcal{H}_r^{(1)}=\sum_{k=0}^{r-1}(-\beta)^kC_{r}^k\sum_{i=1}^N\sum_{j_1\neq \cdots \neq j_k\neq i}
\frac{\partial^{r-k} }{\partial z_i^{r-k}}
\frac{1}{(z_i-z_{j_1})\cdots (z_i-z_{j_k})},
\end{equation}
under the transformation $(-1)^k\frac{\partial^{r-k} }{\partial z^{r-k}}f(z)\rightarrow f(z)\frac{\partial^{r-k} }{\partial z^{r-k}}$, they become
\begin{equation}
\mathcal{\bar H}_r^{(1)}=\sum_{k=0}^{r-1}\beta^kC_{r}^k\sum_{i=1}^N\sum_{j_1\neq \cdots \neq j_k\neq i}
\frac{1}{(z_i-z_{j_1})\cdots (z_i-z_{j_k})}\frac{\partial^{r-k} }{\partial z_i^{r-k}}.
\end{equation}
The similarity transformation $\Delta^{\beta}\circ\mathcal{\bar H}_r^{(1)}\circ\Delta^{-\beta}$ gives the rational Calogero-Sutherland Hamiltonians \cite{2306.06623,2301.11877,2303.05273}.

Unlike the case of the Hermitian matrix model, we observe that the extended operators of (\ref{Hn1})
\begin{equation}\label{gHn1}
\sum_{k=0}^{r-1}(-\beta)^kC_{r}^k\sum_{i=1}^N\sum_{j_1\neq \cdots \neq j_k\neq i}
\frac{\partial^{r-k} }{\partial z_i^{r-k}}
\frac{z_i^{n+r}}{(z_i-z_{j_1})\cdots (z_i-z_{j_k})}
\end{equation}
do not give rise to constraints for the $\beta$-deformed matrix model (\ref{DPF}).

To give the higher order constraints for (\ref{DPF}), let us construct the $\beta$-deformed operators
\begin{align}\label{dwnr}
{\mathcal W}_{n}^{r}&=\sum_{k=0}^{r-1}(-1)^kC_{r-1}^k
\sum_{i=1}^N\frac{\partial^{r-k} }{\partial z_i^{r-k}}z_i^{n+r}
\Delta^{-2\beta}
\frac{\partial^{k} \Delta^{2\beta}}{\partial z_i^{k}}\nonumber\\
&=\sum_{k=0}^{r-1}(-1)^kC_{r-1}^k
\sum_{i=1}^N\frac{\partial^{r-k} }{\partial z_i^{r-k}}z_i^{n+r}
\sum_{\lambda=(\lambda_1,\cdots,\lambda_l)\mapsto k}
\frac{C(\lambda)}{z_{\lambda}}\nonumber\\
&\times \sum_{j_1\neq \cdots \neq j_l\neq i}
\frac{k!}{(z_i-z_{j_1})^{\lambda_1}\cdots (z_i-z_{j_l})^{\lambda_l}},
\end{align}
where $r\geq 1$, $n\geq -r$, $\lambda=(\lambda_1,\cdots,\lambda_l)$, $\lambda_1\geq\cdots \geq\lambda_l\geq 1$ is a partition of $k$,
$C(\lambda)=C(\lambda_1)\cdots C(\lambda_l)$ and $C(\lambda_i)=2\beta(2\beta-1)\cdots(2\beta-\lambda_i+1)$,
$z_{\lambda}=\lambda_1!\cdots \lambda_l!\prod_{k}m_k!$ and $m_k=|\{i|\lambda_i=k\}|$.
When $\beta=1$, (\ref{dwnr}) reduce to the total derivative operators (\ref{anwnr}).

The first few members of (\ref{dwnr}) are as follows
\begin{eqnarray}
\mathcal{W}_n^1&=&\sum_{i=1}^{N}\frac{\partial}{\partial z_i}z^{n+1}_i,\nonumber\\
\mathcal{W}_n^2&=&\sum_{i=1}^{N}\frac{\partial^2}{\partial z^2_i}z^{n+2}_i
-2\beta\sum_{i=1}^{N}\sum_{j\neq i}\frac{\partial}{\partial z_i}\frac{z^{n+2}_i}{z_i-z_j},\nonumber\\
{\mathcal W}_{n}^{3}&=&\sum_{i=1}^N\frac{\partial^3 }{\partial z_i^3}z_i^{n+3}
-4\beta\sum_{i=1}^N\sum_{j\neq i}\frac{\partial^2 }{\partial z_i^2}\frac{z_i^{n+3}}{z_i-z_j}+4\beta^2\sum_{k\neq j\neq i}\frac{\partial }
{\partial z_i}\frac{z_i^{n+3}}{(z_i-z_j)(z_i-z_k)}\nonumber\\
&&+2\beta(2\beta-1)\sum_{j\neq i}\frac{\partial }{\partial z_i}\frac{z_i^{n+3}}{(z_i-z_j)^2}.
\end{eqnarray}

In similarity with (\ref{wnr1}), the $\beta$-deformed operators (\ref{dwnr}) can be rewritten as
\begin{eqnarray}\label{drepn1}
\mathcal{W}_{n}^{r}=\sum_{i=1}^N\frac{\partial }{\partial z_i}{\bar D}^{r-1}_iz^{n+r}_i,
\end{eqnarray}
where $\bar D_i=\frac{\partial }{\partial z_i}-2\beta\sum_{j\neq i}\frac{1}{z_i-z_j}$.

Furthermore, when $n\neq -r$ and $r\geq 2$, there are the expressions
\begin{eqnarray}\label{drepn2}
\mathcal{W}_{n}^{r}
=\sum_{k=0}^{r-2}C_{r-1}^kA_{n+r}^k\sum_{i=1}^NL_{n+r-k,i}{\bar D}^{r-1-k}_i
+A_{n+r}^{r-1}{\mathcal W}_{n}^{1},
\end{eqnarray}
where $L_{n+r-k,i}=\frac{\partial }{\partial z_i}z^{n+r-k}_i$.

Since the operators $\bar D_i$ satisfy  $\bar D_i\Delta^{2\beta}=0$, we have the equalities for the power of the Vandermonde determinant from (\ref{drepn1}) and (\ref{drepn2})
\begin{eqnarray}
\mathcal{W}_{-r}^{r}\Delta^{2\beta}&=&0,\ \ r\in  \mathbb{Z}_{+},\nonumber\\
(\mathcal{W}_{n}^{r}-A_{n+r}^{r-1}{\mathcal W}_{n}^{1})\Delta^{2\beta}&=&0,\ \ n\neq -r.
\end{eqnarray}

In addition, for the operators $\mathcal{\bar W}_{n}^{r}=\mathcal{W}_{n}^{r}+{\mathcal W}_{0}^{1}$, we have
\begin{eqnarray}
\mathcal{\bar W}_{0}^{r}\Delta^{2\beta}&=&(r!+1)N(\beta(N-1)+1)\Delta^{2\beta},\nonumber\\
\mathcal{\bar W}_{n}^{r}\varphi_{n,r}&=&N(\beta(N-1)+1)\varphi_{n,r}, \ \ n\neq 0,
\end{eqnarray}
where $\varphi_{n,r}={\rm exp}(-\frac{1}{n}\mathcal{W}_{n}^{r})\Delta^{2\beta}$, and the commutators
\begin{equation}
[\mathcal{W}_{n}^{r},{\mathcal W}_{0}^{1}]=-n\mathcal{W}_{n}^{r}
\end{equation}
have been used.

By inserting ${\mathcal W}_{n}^{r}$ (\ref{dwnr}) into the partition function (\ref{DPF}) and using
\begin{eqnarray}
{\mathcal W}_{n}^{r}(\Delta^{2\beta}e^V)&=&\left(\sum_{k=0}^rC_r^kA_{n+r}^k
\sum_{i=1}^Nz_i^{n+r-k}Q_{r-k}[\partial_{z_i}V]\right.\nonumber\\
&&\left.+2\beta\sum_{k=0}^{r-1}C_{r-1}^kA_{n+r}^k\sum_{j\neq i}\frac{z_i^{n+r-k}Q_{r-k-1}[\partial_{z_i}V]}{z_i-z_j}\right)\Delta^{2\beta}e^V,
\end{eqnarray}
we obtain a series of constraints
\begin{equation}\label{wcons3}
{\mathcal {\hat W}}_{n}^{r}Z_{\beta}=0, \ \ n\geq -r,
\end{equation}
where ${\mathcal{\hat W}}_{n}^{r}$ are reducible to the Virasoro constraint operators (\ref{defvirasoro})
\begin{equation}\label{dWnrcons}
{\mathcal{\hat W}}_{n}^{r}=\sum_{k=0}^{r-1}C_{r-1}^kA_{n+r}^k{\tilde \rho}(P^{n+r-k}Q_{r-k-1}[j(P)]),
\end{equation}
${\tilde \rho}(P^{n+1})={\mathcal {\hat L}}_n$.

\section{Conclusion}
We have constructed three sets of higher order total derivative operators $H_n^{(m)}$ (\ref{Hnm}), $W_n^r$ (\ref{deroperators}) and $\mathbb{W}_{n}^{r}$ (\ref{anwnr}).
By inserting them into the integrand of the partition function for the Hermitian matrix model, the higher order constraints (\ref{Hnmcons}), (\ref{wcons1}) and (\ref{wcons2})
were obtained.  Through rescaling variable transformations for the constraint operators $\hat H_{-n}^{(m)}$ (\ref{adhatw-rr}),
the remarkable property is that these constraint operators  become the $W$-operators of $W$-representations for the partition function hierarchies in Refs. \cite{wangliu,2301.04107}.
For the constraint operators $\hat W_n^r$ in (\ref{wcons1}), we proved that they are reducible to the Virasoro constraint operators (\ref{virconop}). Using the operators $\hat W_n^r$,
the Itoyama-Matsuo conjecture has been proved. For the constraint operators ${\mathbb{\hat W}}_{n}^{r}$ (\ref{wconsop2}), we have presented their expressions in terms of the Virasoro constraint operators (\ref{virconop}).

We have also constructed the $\beta$-deformed higher order total derivative operators $\mathcal{H}_n^{(m)}$ (\ref{dHnm}) and ${\mathcal W}_{n}^{r}$ (\ref{dwnr}), and derived the constraints for the $\beta$-deformed Hermitian matrix model. Similarly, by rescaling variable transformations, we found that the constraint operators $\mathcal{\hat H}_{-n}^{(m)}$ (\ref{dhatHnm}) give the $W$-operators of $W$-representations
for the $\beta$-deformed partition function hierarchies in Ref. \cite{2301.12763}. The expressions of $\beta$-deformed constraint operators ${\mathcal {\hat W}}_{n}^{r}$ (\ref{dWnrcons}) by Virasoro constraint operators (\ref{defvirasoro})
have been provided. For further research, it is worthwhile to investigate the higher order constraints for multi-matrix models.
\section*{Acknowledgements}
I am grateful to Wei-Zhong Zhao, Fan Liu, Jie Yang and Min-Li Li for their helpful discussions. This work is supported by the National Natural Science Foundation of China (No. 12205368)
and the Fundamental Research Funds for the Central Universities, China (No. 2024ZKPYLX01).

\end{document}